\documentclass[letterpaper]{article}

\newcommand{\FULLPAPER}

\def\year{2017}\relax
\usepackage{aaai17}
\usepackage{times}
\usepackage{helvet}
\usepackage{courier}
\ifdefined\ANONY
\else
\fi
\ifdefined\ANONY
\else
    \author{
        Francesco Ald\`a \\
        Horst G\"ortz Institute for IT Security \\
        and Faculty of Mathematics\\
        Ruhr-Universit\"at Bochum\\
        Germany\\
        \texttt{francesco.alda@rub.de} \\
        \And
        Benjamin I.P. Rubinstein\\
        Dept. Computing and Information Systems\\
        The University of Melbourne\\
        Australia\\
        \texttt{brubinstein@unimelb.edu.au} \\
    }
\fi

\usepackage[utf8]{inputenc} 
\usepackage{booktabs}       
\usepackage{amsfonts}       
\usepackage{nicefrac}       
\usepackage{microtype}      

\usepackage{xspace}
\usepackage{algpseudocode}
\usepackage{algorithm}
\usepackage[english]{babel}
\usepackage{amssymb}
\usepackage{amsmath}
\usepackage{amsthm}
\usepackage{mathtools}
\usepackage{enumerate}

\newcommand\filltoend{\leavevmode{\unskip
  \leaders\hrule height.5ex depth\dimexpr-.5ex+0.4pt\hfill\hbox{}%
  \parfillskip=0pt\endgraf}}

\newcommand{\iffullelse}[2]{\ifdefined\FULLPAPER{#1}\else{#2}\fi}
\newcommand{\cf}{\emph{cf.}\xspace}
\newcommand{\ie}{\emph{i.e.,}\xspace}
\newcommand{\eg}{\emph{e.g.,}\xspace}
\newcommand{\str}{\ensuremath{\star}}
\renewcommand{\l}{\ell}
\newcommand{\cD}{\ensuremath{\mathcal{D}}\xspace}

\newcommand{\cX}{\ensuremath{\mathcal{X}}\xspace}
\newcommand{\cY}{\ensuremath{\mathcal{Y}}\xspace}
\newcommand{\cM}{\mathcal{M}}
\newcommand{\cA}{\mathcal{A}}
\newcommand{\cR}{\mathcal{R}}

\newcommand{\reals}{\mathbb{R}}
\newcommand{\eps}{\varepsilon}
\newcommand{\Lap}{\mathrm{Lap}}

\newcommand{\nats}{\mathbb{N}}
\newcommand{\ra}{\rightarrow}
\newcommand{\cH}{\mathcal{H}}
\newcommand{\vect}[1]{\boldsymbol{#1}}
\DeclarePairedDelimiter{\floor}{\lfloor}{\rfloor}
\DeclareMathOperator*{\argmin}{arg\,min}
\newcommand{\citenoun}[1]{\citeauthor{#1}~\shortcite{#1}\xspace} 
\newcommand{\citenaked}[1]{\citeauthor{#1}~\citeyear{#1}\xspace} 
\renewcommand{\Pr}{\mathbb{P}}

\newtheorem{theorem}{Theorem}

\newtheorem{lemma}[theorem]{Lemma}
\newtheorem{proposition}[theorem]{Proposition}

\newtheorem{definition}{Definition}

\title{The Bernstein Mechanism:\\ Function Release under Differential Privacy
\ifdefined\SUBTITLE
  \\ \SUBTITLE
\fi
}

\begin{document}

\nocopyright

\maketitle

\begin{abstract}
We address the problem of general function release under differential privacy,
by developing a functional mechanism that applies
under the weak assumptions of oracle access to target function evaluation and sensitivity. These conditions permit treatment of functions described explicitly or implicitly as algorithmic black boxes.
We achieve this result by leveraging the iterated Bernstein operator for polynomial
approximation of the target function, and polynomial coefficient perturbation. Under weak regularity conditions, we establish fast rates on utility measured by high-probability uniform
approximation. We provide a lower bound on the utility achievable for any functional
mechanism that is $\eps$-differentially private. The generality of our mechanism is
demonstrated by the analysis of a number of example learners, including naive Bayes, non-parametric
estimators and regularized empirical risk minimization.
Competitive rates are demonstrated for kernel density estimation; and
$\eps$-differential privacy is achieved for a broader class of support vector machines
than known previously.
\end{abstract}

\section{Introduction}\label{sec:introduction}

In recent years, {\em differential privacy}~\cite{DMNS2006} has emerged as a leading paradigm for privacy-preserving statistical analyses. It provides formal guarantees that aggregate statistics output by a randomized mechanism are not significantly influenced by the presence or absence of an individual datum.
\emph{Where the Laplace mechanism~\cite{DMNS2006} is a de facto approach for converting vector-valued functions to differential privacy, in this paper we seek an equivalent approach for privatizing function-valued mappings.} We achieve our goal through the development of a novel Bernstein functional mechanism.
Unlike existing mechanisms, ours applies to releasing explicitly and implicitly defined functions, and is characterized by a full theoretical analysis.

Our setting is the release of functions that depend on privacy-sensitive training data, and that can be subsequently evaluated on arbitrary test points. This non-interactive setting matches a wide variety of learning tasks from naive Bayes classification, non-parametric methods (kernel density estimation and regression) where the function of train and test data is explicit, to generalized linear models, support vector machines where the function is only implicitly defined by an iterative algorithm. Our generic mechanism is based on functional approximation by Bernstein basis polynomials, specifically via an iterated Bernstein operator. Privacy is guaranteed by sanitizing the coefficients of approximation, which requires only function evaluation. It is the very limited oracle access required by our mechanism---to non-private function evaluation and sensitivity---that grants it broad applicability akin to the Laplace mechanism.

The Bernstein polynomials central to our mechanism are used in the Stone-Weierstrass theorem to uniformly approximate any continuous function on a closed interval.
Moreover, the Bernstein operator offers several advantages such as data-independent bounds, no requirement of access to target function derivatives, and yields approximations that are pointwise convex combinations of the function evaluations on a cover. As a result, applying privacy-preserving perturbations to the approximation's coefficients permits us to control utility and achieve fast convergence rates.

In addition to being analyzed in full, the Bernstein mechanism is easy to use. We demonstrate this with a variety of example analyses of the mechanism applied to learners.
Finally, we provide a lower bound that fundamentally limits utility under private function release, partly resolving a question posed by \citenoun{HRW2013}. This matches (up to logarithmic factors) our upper bound in the linear case.

\paragraph*{Related Work.} Polynomial approximation has proven useful in differential privacy outside function release~\cite{thaler2012faster,chandrasekaran2014faster}. Few previous attempts have been made towards private function release. 
\citenoun{HRW2013} add Gaussian process noise which only yields a weaker form of privacy, namely $(\eps,\delta)$-differential privacy, and does not admit general utility rates. \citenoun{zhang2012functional} introduce a functional mechanism for the more specific task of perturbing the objective in private optimization, but they assume separability in the training data and do not obtain rates on utility.

\citenoun{WFZW2013} propose a mechanism that releases a summary of data in a trigonometric basis, able to respond to queries that are smooth as in our setting, but are also required to be separable in the training dataset as assumed by~\citenoun{zhang2012functional}. A natural application is kernel density estimation, which would achieve a rate of $O\left(\log(1/\beta)/(n\eps)\right)^{h/(\l+h)}$ as does our approach. Private KDE has also been explored in various other settings~\cite{duchi2013local} and under weaker notions of utility~\cite{HRW2013}.
\citenoun{ZRD2016} explore discrete naive Bayes under differential privacy, while we investigate parametric Gaussian and non-parametric KDE for class-conditional likelihoods.

As an example of an implicitly defined function, we consider regularized empirical risk minimization such as logistic regression, ridge regression, and the SVM. Previous mechanisms for private SVM release and ERM more generally~\cite{CM2008,RBHT2012,chaudhuri2011differentially,jain2014near,jain2013differentially,bassily2014private} require finite-dimensional feature mapping or translation-invariant kernels. \citenoun{HRW2013} consider more general mappings but provide $(\eps,\delta)$-differential privacy. Our treatment of regularized ERM extends to kernels that may be translation-variant with infinite-dimensional mappings, while providing stronger privacy guarantees.

\section{Preliminaries}

\paragraph*{Notation and Problem Setting.}

Throughout the paper, vectors are written in bold and the $i$-th component of a vector $\vect{x}$ is denoted by $x_i$.
We consider $\cX$ an arbitrary (possibly infinite) domain and $\cD \in \cX^n$ a database of $n$ points in $\cX$. We refer to $n$ as the size of the database $\cD$. For a positive integer $\l$, let $\cY=[0,1]^{\l}$ be a set of query points and $F\colon \cX^n\times\cY\ra\reals$ the target function. Once the database $\cD$ is fixed, we denote by $F_{\cD} = F(\cD,\cdot)$ the function parameterized by~$\cD$ that we aim to release. For example: \cD might represent a training set---over \cX a product space of feature vectors and labels---with \cY representing test points from the same feature space; $F_{\cD}$ would then be a classifier resulting from training on \cD. Section~\ref{sec:examples} presents examples for $F$. In Section~\ref{sec:mechanism}, we show how to privately release the function $F_{\cD}$ and we provide alternative error bounds depending on the regularity of $F$.

\begin{definition}\label{def:smoothness}
    Let $h$ be a positive integer and $T>0$. A function $f\colon [0,1]^{\l}\ra\reals$ is $(h,T)$-smooth if it is $C^h([0,1]^{\l})$ and its partial derivatives up to order $h$ are all bounded by~$T$.
\end{definition}

\begin{definition}\label{def:hoelder}
    Let $0\leq \gamma\leq 1$ and $L>0$. A function $f\colon [0,1]^{\l}\ra\reals$ is $(\gamma,L)$-H\"older continuous if, for every $\vect{x},\vect{y}\in [0,1]^{\l}$, $|f(\vect{x})-f(\vect{y})|\leq L \|\vect{x}-\vect{y}\|_{\infty}^\gamma$.
    When $\gamma=1$, we refer to $f$ as $L$-Lipschitz.
\end{definition}

Our goal is to develop a private release mechanism for the function $F_{\cD}$ in the \emph{non-interactive} setting. A non-interactive mechanism takes a function $F$ and a database $\cD$ as inputs and outputs a synopsis $\cA$ which can be used to evaluate the function $F_{\cD}$ on $\cY$ without accessing the database $\cD$ further.

\paragraph*{Differential Privacy.}

To provide strong privacy guarantees on the release of $F_{\cD}$, we adopt the well-established notion of differential privacy.

\begin{definition}[\citenaked{DMNS2006}]\label{def:diffprivacy}
	Let $\cR$ be a (possibly infinite) set of responses. A mechanism $\cM:\cX^\str \ra \cR$ (meaning that, for every $\cD\in\cX^\str=\bigcup_{n>0}\cX^n$, $\cM(\cD)$ is an $\cR$-valued random variable) is said to provide {\em $(\eps,\delta)$-differential privacy} for $\eps>0$ and $0\leq\delta< 1$ if, for every $n\in\nats$, for every pair $(\cD,\cD') \in \cX^n\times\cX^n$ of databases differing in one entry only (henceforth denoted by $\cD\sim\cD'$), and for every measurable $S \subseteq \cR$, we have $\Pr[\cM(\cD) \in S] \le e^\eps\Pr[\cM(\cD') \in S] + \delta.$
    If $\delta=0$ we simply say that $\cM$ provides $\eps$-differential privacy.
\end{definition}

By limiting the influence of data on the induced response distribution, a powerful adversary (with knowledge of all but one input datum, the mechanism up to random source, and unbounded computation) cannot effectively identify an unknown input datum from mechanism responses.
The Laplace mechanism~\cite{DMNS2006} is a generic tool for differential privacy: adding zero-mean Laplace noise\footnote{A $\Lap(\lambda)$-distributed real random variable $Z$ has probability density proportional to $\exp(-|y|/\lambda)$.} to a vector-valued function provides privacy if the noise is calibrated to the function's sensitivity.

\begin{definition}[\citenaked{DMNS2006}]\label{def:sensitivity}
    The {\em sensitivity} of a function $f:\cX^n\ra\reals^d$ is given by $ S(f) = \sup_{\cD\sim\cD'}\|f(\cD)-f(\cD')\|_1$,
    where the supremum is taken over all $\cD,\cD' \in \cX^n$ that differ in one entry only. The sensitivity of a function $F\colon\cX^n\times \cY\ra\reals^d$ is defined as $ S(F) = \sup_{\vect{y}\in \cY}S(F(\cdot,\vect{y}))$.
\end{definition}

\begin{lemma}[\citenaked{DMNS2006}]\label{lem:sensitivity}
    Let $f:\cX^n\ra\reals^d$ be a non-private function of finite sensitivity, and let $\vect{Z}\sim\Lap(S(f)/\eps)^d$. Then, the random function $\tilde{f}(\cD) = f(\cD)+\vect{Z}$ provides $\eps$-differential privacy.
\end{lemma}

Given a mechanism, we measure its accuracy as follows.

\begin{definition}\label{def:accuracy}
    Let $F\colon\cX^n\times \cY\ra\reals$. A mechanism $\cM$ is {\em $(\alpha,\beta)$-accurate} with respect to $F_{\cD}$ if for any database $\cD\in\cX^n$ and $\cA=\cM(\cD)$, with probability at least $1-\beta$ over the randomness of $\cM$, $\sup_{\vect{y}\in\cY} |\cA(\vect{y}) - F_{\cD}(\vect{y})|\leq\alpha.$
\end{definition}

\section{The Bernstein Mechanism}\label{sec:mechanism}

Algorithm~\ref{alg:bernstein_mechanism} introduces a differentially-private mechanism for releasing $F_\cD: \cY\ra \reals$, a family of $(h,T)$-smooth or $(\gamma,L)$-H\"older continuous functions, parameterized by $\cD\in\cX^n$.

\begin{algorithm}[H]
    \begin{algorithmic}[1]
        \Statex \textit{Sanitization} -- \textbf{Inputs:} private dataset $\cD\in\cX^{n}$; sensitivity $S(F)$ and oracle access to target $F\colon \cX^n\times\cY\ra\reals$
        \Statex \textbf{Parameters:} cover size $k$, Bernstein order $h$ positive integers; privacy budget $\eps>0$
        \State $P\gets\left(\left\{0,1/k,2/k,\ldots,1\right\}\right)^{\l}$ \Comment{Lattice cover of \cY}
        \State $\lambda \gets S(F)(k+1)^{\l}/\eps$ \Comment{Perturbation scale}
        \State For each $\vect{p}=(p_1,\ldots,p_{\l})\in P$:
        \State \hspace{2em} $\widetilde{F_{\cD}}(\vect{p}) \gets F_{\cD}(\vect{p}) + Z$, where $Z\stackrel{i.i.d.}{\sim}\Lap(\lambda)$
        \State \textbf{Return:} $\left\{\widetilde{F_{\cD}}(\vect{p})\mid \vect{p}\in P\right\}$
        \Statex \filltoend
        \Statex \textit{Evaluation} -- \textbf{Inputs:} query $\vect{y}\in\cY$; private response $\left\{\widetilde{F_{\cD}}(\vect{p})\mid \vect{p}\in P\right\}$
        \State $b^{(h)}_{\nu_i,k}\gets$
        Compute basis \Comment{See Definition~\ref{def:iteretatedbernsteinpoly}}
        \State \textbf{Return:} $ \sum_{j=1}^{\l}\sum_{\nu_j=0}^{k}\widetilde{F_{\cD}}\left(\frac{\nu_1}{k},\cdots,\frac{\nu_{\l}}{k}\right)\prod_{i=1}^{\l}{b^{(h)}_{\nu_i,k}(y_i)}$
    \end{algorithmic}
    \caption{The Bernstein mechanism~\label{alg:bernstein_mechanism}}
\end{algorithm}

The mechanism makes use of the iterated Bernstein polynomial of $F_{\cD}$, which we introduce
next (for a comprehensive survey refer to~\citenaked{L1953}, \citenaked{M1973}). This approximation consists of a linear combination of so-called Bernstein basis polynomials, whose coefficients are evaluations of target $F_{\cD}$ on a (lattice) cover $P$.

We briefly introduce the univariate Bernstein basis polynomials and state some of their properties.

\begin{definition}\label{def:bernsteinbasis}
    Let $k$ be a positive integer. The {\em Bernstein basis polynomials} of degree $k$ are
    defined as $b_{\nu,k}(y) = \binom{k}{\nu}y^{\nu}(1-y)^{k-\nu}$ for $\nu=0,\ldots,k$.
\end{definition}

\begin{proposition}[\citenaked{L1953}]\label{prop:bernsteinproperties}
    For every $y\in[0,1]$, any positive integer $k$ and $0\leq\nu\leq k$,
    we have $b_{\nu,k}(y)\geq 0$ and $\sum_{\nu=0}^{k} b_{\nu,k}(y)=1$.
\end{proposition}

In order to introduce the iterated Bernstein polynomials, we first need to recall the Bernstein operator.

\begin{definition}\label{def:bernsteinpoly}
    Let $f\colon [0,1]\ra \reals$ and $k$ be a positive integer. The {\em Bernstein polynomial} of $f$ of degree $k$ is defined as
    $B_{k}(f;y) = \sum_{\nu=0}^{k}f\left(\nu/k\right)b_{\nu,k}(y)$.
\end{definition}

The \emph{Bernstein operator} $B_{k}$ maps a function $f$, defined on $[0,1]$, to $B_{k} f$, where the function $B_{k} f$ evaluated at $y$ is $B_k(f;y)$. Note that the Bernstein operator is linear and if $f(y)\in[a_1,a_2]$ for every $y\in[0,1]$, then from Proposition~\ref{prop:bernsteinproperties} it follows that $B_{k}(f;y)\in[a_1,a_2]$ for every positive integer $k$ and $y\in[0,1]$. Moreover, it is not hard to show that any linear function is a fixed point for $B_{k}$\iffullelse{. For completeness, we provide a short proof in Appendix~\ref{sec:fixedpoints}.}{ (\cf the full report~\citenaked{AR2016}).}

\begin{definition}[\citenaked{M1973}]\label{def:iteretatedbernsteinpoly}
    Let $h$ be a positive integer. The \emph{iterated Bernstein operator} of order $h$ is defined as the sequence of linear operators $B^{(h)}_{k} = I-(I-B_k)^{h} = \sum_{i=1}^{h}{h \choose i}(-1)^{i-1}B^{i}_{k}$,
    where $I=B^{0}_{k}$ denotes the identity operator and $B^{i}_{k}$ is defined inductively as $B^{i}_{k}=B_{k}\circ B^{i-1}_{k}$ for $i\geq 1$. The \emph{iterated Bernstein polynomial} of order $h$ can then be computed as:
    \[B^{(h)}_{k}(f;y) = \sum_{\nu=0}^{k}f\left(\frac{\nu}{k}\right)b^{(h)}_{\nu,k}(y)\ ,\]
    where $b^{(h)}_{\nu,k}(y) = \sum_{i=1}^{h}{h \choose i}(-1)^{i-1}B^{i-1}_{k}(b_{\nu,k};y)\ .$
\end{definition}

We observe that $B^{(1)}_{k} = B_{k}$. Although the bases $b^{(h)}_{\nu,k}$ are not always positive for $h\geq 2$, we still have $\sum_{\nu=0}^{k}b^{(h)}_{\nu,k}(y)=1$ for every $y\in[0,1]$.
The iterated Bernstein polynomial of a multivariate function $f\colon [0,1]^{\l}\ra \reals$ is analogously defined.

\begin{definition}\label{def:iteratedbernsteinpolymultivariate}
    Assume $f\colon [0,1]^{\l}\ra \reals$ and let $k_1,\ldots,k_{\l},h$ be positive integers. The (multivariate) iterated Bernstein polynomial of $f$ (of order $h$) is defined as
    \begin{equation*}
        B^{(h)}_{k_1,\ldots,k_{\l}}(f;\vect{y}) = \sum_{j=1}^{\l}\sum_{\nu_j=0}^{k_j}f\left(\frac{\nu_1}{k_1},\cdots,\frac{\nu_{\l}}{k_{\l}}\right)\prod_{i=1}^{\l}{b^{(h)}_{\nu_i,k_i}(y_i)}.
    \end{equation*}
\end{definition}

For ease of exposition, we fix user-selected $k\in\mathbb{N}$ such that $k_1=\ldots=k_{\l}=k$. The Bernstein mechanism perturbs the evaluation of $F_{\cD}$ on a lattice cover $P$ of $\cY=[0,1]^{\l}$ parameterized by $k$.

\section{Analysis of Mechanism Privacy and Utility}

In the following result, we assume $\l$ to be an arbitrary but fixed constant with $\cY = [0,1]^{\l}$.
We underline that this is a common assumption in the differential privacy literature, especially when dealing with Euclidean spaces~\cite{BLR2008,dwork2009differential,wasserman2010statistical,lei2011differentially,WFZW2013}.

\begin{theorem}[Main Theorem]\label{thm:maintheorem}
    Let $\l,h\in\mathbb{N_{+}}$, $0<\gamma\leq 1$, $L>0$ and $T>0$ be constants. Let $\cX$ be an arbitrary space and $\cY = [0,1]^{\l}$. Let furthermore $F\colon \cX^{n}\times \cY\ra\reals$ with $S(F) = o(1)$. For $\eps>0$, the Bernstein mechanism $\cM$ provides $\eps$-differential privacy. Moreover, for $0<\beta<1$ the mechanism $\cM$ is $(\alpha,\beta)$-accurate with error scaling as follows, where hidden constants depend on $\l,L,\gamma,T,h$.
    \begin{enumerate}[(i)]
        \item
        If $F_{\cD}$ is $(2h,T)$-smooth for every $\cD\in\cX^{n}$,
        there exists $k=k(S(F),\eps,\beta,\l,h,T)$ such that $\alpha = O\left(\frac{S(F)}{\eps}\log(1/\beta)\right)^{\frac{h}{\l + h}}$;
        \item
        If $F_{\cD}$ is $(\gamma,L)$-H\"older continuous for every $\cD\in\cX^{n}$,
        there exists $k=k(S(F),\eps,\beta,\l,\gamma,L)$ such that $\alpha = O\left(\frac{S(F)}{\eps}\log(1/\beta)\right)^{\frac{\gamma}{2\l + \gamma}}$; and
        \item
        If $F_{\cD}$ is linear for every $\cD\in\cX^{n}$,
        there exists a constant $k$ such that $\alpha = O\left(\frac{S(F)}{\eps}\log(1/\beta)\right)$.
    \end{enumerate}
    Moreover, if $1/S(F)\leq\mathrm{poly}(n)$, then the running-time of the mechanism and the running-time per evaluation are both polynomial in $n$ and $1/\eps$.
\end{theorem}

\subsection{Proof of the Main Theorem\footnote{For sake of clarity, in \iffullelse{Appendix~\ref{sec:onedimproof}}{the full report~\citenoun{AR2016}} we provide a self-contained proof of Theorem~\ref{thm:maintheorem} for $\l=1$. Although it is not a prerequisite to the general result, it reflects the building blocks used in this section.}}\label{sec:mainproof}

To prove privacy we note that only the coefficients of the Bernstein polynomial of $F_{\cD}$ are sensitive and need to be protected. In order to provide $\eps$-differential privacy, these coefficients---evaluations of target $F_{\cD}$ on a cover---are perturbed by means of Lemma~\ref{lem:sensitivity}. In this way, we can release the sanitized coefficients and use them for unlimited, efficient evaluation of the approximation of $F_{\cD}$ over $\cY$, without further access to the data $\cD$. To establish utility, we separately analyze error due to the polynomial approximation of $F_{\cD}$ and error due to perturbation. 

In order to analyze the accuracy of our mechanism, we denote by $\widetilde{B^{(h)}_{k}}(F_{\cD};\vect{y})$ the iterated Bernstein polynomial of order $h$ constructed using the coefficients output by the mechanism $\cM$. The error $\alpha$ introduced by the mechanism can be expressed as follows:
\begin{align}
    \alpha &= \max_{\vect{y}\in[0,1]^{\l}}\left|F_{\cD}(\vect{y}) - \widetilde{B^{(h)}_{k}}(F_{\cD};\vect{y})\right| \label{eq:error1main} \\
    &\leq \max_{\vect{y}\in[0,1]^{\l}}\left|\widetilde{B^{(h)}_{k}}(F_{\cD};\vect{y}) - B^{(h)}_{k}(F_{\cD};\vect{y})\right|\label{eq:error2main}\\
    \ &\mathrel{\phantom{\leq\ }}   
    \mathrel + \max_{\vect{y}\in[0,1]^{\l}}\left|F_{\cD}(\vect{y}) - B^{(h)}_{k}(F_{\cD};\vect{y})\right|\nonumber.
\end{align}

For every $\vect{y}\in [0,1]^{\l}$, the first summand in Equation~\eqref{eq:error2main} consists of the absolute value of an affine combination of independent Laplace-distributed random variables.

\begin{proposition}\label{prop:affinelaplacemultivariate}
    Let $\Gamma=\{\vect{\nu}\in\mathbb{N}^{\l}\mid 0\leq \nu_j\leq k \text{ for } 1\leq j\leq \l\}$. For every $\vect{\nu}=(\nu_1,\ldots,\nu_{\l})\in\Gamma$ let $Z_{\vect{\nu}}\stackrel{i.i.d.}{\sim} \Lap(\lambda)$. Moreover, let $\tau\geq0$ and constant $C_{h,\l}$ depend only on $h,\l$. Then:
    \[
    \Pr\left[\max_{\vect{y}\in[0,1]^{\l}}\left|\sum_{j=1}^{\l}\sum_{\nu_j=0}^{k}Z_{\vect{\nu}}\prod_{i=1}^{\l}{b^{(h)}_{\nu_i,k}(y_i)}\right| \geq \tau\right] \leq
    e^{-\tau/(C_{h,\l}\lambda)}.
    \]
\end{proposition}

The proof of Proposition~\ref{prop:affinelaplacemultivariate} follows from a result of~\citenoun{P1965} on the concentration of convex combinations of random variables drawn i.i.d. from a log-concave symmetric distribution\iffullelse{. For completeness, we give the full proof in Appendix~\ref{sec:affinelaplacemultivariate}.}{ (\cf the full report~\citenaked{AR2016}).}
Proposition~\ref{prop:affinelaplacemultivariate} implies that with probability at least $1-\beta$ the first summand in Equation~\eqref{eq:error2main} is bounded by $O\left(S(F)k^{\l}\log(1/\beta)/\eps\right)$.
In order to bound the second summand we make use of the following (unidimensional) convergence rates.

\begin{theorem}[\citenaked{M1973}]\label{thm:bernstainapproxsmooth}
    Let $h$ be a positive integer and $T>0$. If $f\colon [0,1]\ra \reals$ is a $(2h,T)$-smooth function,
    then, for all positive integers $k$ and $y\in[0,1]$, $\left|f(y)-B^{(h)}_{k}(f;y)\right|\leq TD_hk^{-h}$,
    where $D_h$ is a constant independent of $k,f$ and $y\in[0,1]$.
\end{theorem}

\begin{theorem}[\citenaked{K1938}; \citenaked{M1999}]\label{thm:bernstainapproxhoelder}
    Let $0<\gamma\leq 1$ and $L>0$. If $f\colon [0,1]\ra \reals$ is a $(\gamma,L)$-H\"older continuous function,
    then $\left|f(y)-B^{(1)}_{k}(f;y)\right|\leq L\left(4k\right)^{-\gamma/2}$ for all positive integers $k$ and $y\in[0,1]$.
\end{theorem}

By induction, it is possible to show that the approximation error of the multivariate iterated Bernstein polynomial can be bounded by $O(\l g(k))=O(g(k))$, if the error of the corresponding univariate polynomial is bounded by $g(k)$\iffullelse{. For completeness, we provide a proof in Appendix~\ref{sec:induction}.}{ (\cf the full report~\citenaked{AR2016}).}

All in all, the error $\alpha$ introduced by the mechanism can thus be bounded by
\begin{equation}\label{eq:error4main}
    \alpha= O\left(g(k) + \frac{S(F)k^{\l}}{\eps}\log(1/\beta)\right).
\end{equation}

Since $g(k)$ is a decreasing function in $k$ and the second summand in Equation~\eqref{eq:error4main} is an increasing function in $k$, the optimal value for $k$ (up to a constant factor) is achieved when $k$ satisfies
\begin{equation}\label{eq:valuekmain}
    g(k) = \frac{S(F)k^{\l}}{\eps}\log(1/\beta)\ .
\end{equation}

Solving Equation~\eqref{eq:valuekmain} with the bound for $g(k)$ provided in Theorem~\ref{thm:bernstainapproxsmooth} yields
\[k = \max\left\{1,\left(\frac{\eps}{S(F)\log(1/\beta)}\right)^{\frac{1}{h+\l}}\right\}\]
and substituting the thus obtained value of $k$ into Equation~(\ref{eq:error4main}) yields the first statement of Theorem~\ref{thm:maintheorem}. Similarly, using the bound for $g(k)$ provided in Theorem~\ref{thm:bernstainapproxhoelder} we get the result for H\"older continuous functions. The bound for linear functions follows from the fact that the approximation error is zero for $h=1$ and $k=1$, since linear functions are fixed points of $B_{1}^{(1)}$. Finally, the analysis of the running time follows from observing that, for the optimal cover size $k$ we computed, $k^{\l}$ is always upper bounded by $\eps/(S(F)\log(1/\beta))$ and thus by $\mathrm{poly}(n)$.

\subsection{Discussion}\label{sec:discuss}

\paragraph*{Comparison to Baseline.} Algorithm~\ref{alg:bernstein_mechanism} is based on a relatively simple approach: it evaluates the target function on a lattice cover, adding Laplace noise for privacy. One might be tempted to approximate the input function by rounding a query point $\vect{y}$ to the nearest lattice point $\vect{p}$ and releasing the corresponding noisy evaluation $\widetilde{F_{\cD}}(\vect{p})$. Although it is straightforward to prove that, for $(\gamma,L)$-H\"older continuous functions, such a piecewise constant approximation achieves error $O(1/k^{\gamma})$, this upper bound is essentially tight, as it can be shown by considering the approximation error it achieves for linear functions. Therefore, this method has two main disadvantages: the output function is not even continuous (although we always consider continuous input functions) and for highly smooth input functions it cannot achieve the fast convergence rates of the Bernstein mechanism. In Section~\ref{sec:examples}, we offer further examples supporting this argument.

\paragraph*{$\boldsymbol{(\eps,\delta)}$-Differential Privacy.} We note that our analysis can be easily extended to the relaxed notion of {\em approximate differential privacy} using advanced composition theorems (see for example~\citenaked{DR2014}) instead of sequential composition~\cite{DMNS2006}. Specifically, it suffices to choose the perturbation scale $\lambda_{\delta}=2S(F)\sqrt{2(k+1)^{\l}\log(1/\delta)}/\eps$.

\begin{theorem}
    Let $0<\delta<1$. Under the same assumptions of Theorem~\ref{thm:maintheorem}, the Bernstein mechanism $\cM$ (with perturbation scale $\lambda_{\delta}$) provides $(\eps,\delta)$-differential privacy and is $(\alpha,\beta)$-accurate with error scaling as follows.
    \begin{enumerate}[(i)]
        \item
        If $F_{\cD}$ is $(2h,T)$-smooth for every $\cD\in\cX^{n}$,
        there exists $k=k(S(F),\eps,\delta,\beta,\l,h,T)$ such that $\alpha = O\left(\frac{S(F)}{\eps}\log(1/\beta)\sqrt{\log(1/\delta)}\right)^{\frac{2h}{\l + 2h}}$; and
        \item
        If $F_{\cD}$ is $(\gamma,L)$-H\"older continuous for every $\cD\in\cX^{n}$,
        there exists $k=k(S(F),\eps,\delta,\beta,\l,\gamma,L)$ such that $\alpha = O\left(\frac{S(F)}{\eps}\log(1/\beta)\sqrt{\log(1/\delta)}\right)^{\frac{\gamma}{\l + \gamma}}$.
    \end{enumerate}
\end{theorem}

Even though this relaxation allows for improved accuracy, in this work we explore a different point on the privacy-utility Pareto front and focus our attention on $\eps$-differential privacy, since there is generally a significant motivation for achieving stronger privacy guarantees.
Moreover, to the best of our knowledge, it is unknown whether previous solutions~\cite{HRW2013} even apply to this framework.

\section{Lower Bound}\label{sec:lowerbound}

In this section we present a lower bound on the error that any $\eps$-differentially private mechanism approximating a function $F\colon\cX^{n}\times\cY\ra\reals$ must introduce.

\begin{theorem}\label{thm:lowerbound}
    For $\eps>0$, there exists a function $F\colon\cX^{n}\times\cY\ra\reals$ such that the error that any $\eps$-differentially private mechanism approximating $F$ introduces is $\Omega\left(S(F)/\eps\right)$, with probability arbitrarily close to $1$.
\end{theorem}

\begin{proof}
    In order to prove Theorem~\ref{thm:lowerbound}, we consider $\cX\subset[0,1]^{\l}$ to be a finite set and without loss of generality we view the database $\cD$ as an element of $\cX^{n}$ or as an element of $\mathbb{N}^{|\cX|}$, \ie a histogram over the elements of $\cX$, interchangeably. We can then make use of a general result provided by \citenoun{D2012}.

    \begin{proposition}[\citenaked{D2012}]\label{prop:delowerbound}
    Assume $\cD_1,\cD_2,\ldots,\cD_{2^s} \in\mathbb{N}^{N}$ such that, for every $i$, $\|\cD_i\|_1\leq n$ and, for $i\neq j$, $\|\cD_i-\cD_j\|_1 \leq \Delta$. Moreover, let $f\colon \mathbb{N}^{N}\ra \reals^{t}$ be such that for any $i\neq j$, $\|f(\cD_i)-f(\cD_j)\|_{\infty} \geq \eta$. If $\Delta \leq (s-1)/\eps$, then any mechanism which is $\eps$-differentially private for the query $f$ on databases of size $n$ introduces an error which is $\Omega(\eta)$, with probability arbitrarily close to $1$.
    \end{proposition}

    Therefore, we only need to show that there exists a suitable sequence of databases $\cD_1,\cD_2,\ldots,\cD_{2^s}$, a function $F\colon\cX^{n}\times\cY\ra\reals$ and a $\vect{y}\in\cY$ such that $F(\cdot,\vect{y})$ satisfies the assumptions of Proposition~\ref{prop:delowerbound}. We actually show that this holds for every $\vect{y}\in\cY$. Let $\eps>0$ and $V$ be a non-negative integer. We define $\cX=(\{0,1/(V+8),2/(V+8),\ldots,1\})^{\l}$. Note that $N=|\cX|=(V+9)^{\l}$. Let furthermore $c=\floor*{1/\eps}$ and $n=V+c$. The function $F\colon\cX^{n}\times[0,1]^{\l}\ra\reals$ we consider is defined as follows:
    \[F(\cD,\vect{y}) = \eta(d_0 + \ldots + d_{N-7} + 2d_{N-6} + \ldots + 8d_{N} + \langle \vect{y},\vect{1}\rangle),\]
    where $d_i$ corresponds to the number of entries in $\cD$ whose value is $x_i$, for every $x_i\in\cX$. For $s=3$, we consider the sequence of databases $\cD_1,\cD_2,\ldots,\cD_{8}$, where, for $j\in\{1,2,\ldots,8\}$, $d_i\in\cD_j$ is such that
    \[
        d_i =
        \begin{cases}
            1, & \text{ for } i\in\{0,1,\ldots,V-1\}\\
            c, & \text{ for } i=N-j+8 \\
            0, & \text{ otherwise}
        \end{cases}\ .
    \]
    We first observe that, for every $j\in\{1,2,\ldots,8\}$, $\|\cD_j\|_1 = n$. Moreover, for $i\neq j$, $\|\cD_i-\cD_j\|_1 = 2c \leq 2/\eps$. Finally, for $i\neq j$, $|F(\cD_i,\vect{y})-F(\cD_j,\vect{y})| \geq c\eta$ for every $\vect{y}\in[0,1]^{\l}$. Since $S(F)=7\eta$, Proposition~\ref{prop:delowerbound} implies that, with high probability, any $\eps$-differentially private mechanism approximating $F$ must introduce an error of order $\Omega\left(S(F)/\eps\right)$.
\end{proof}

\section{Examples}\label{sec:examples}

In this section, we demonstrate the versatility of the Bernstein mechanism through the analysis of a range of example learners.

\begin{figure}[t!]
    \centering
    \includegraphics[width=0.9\columnwidth]{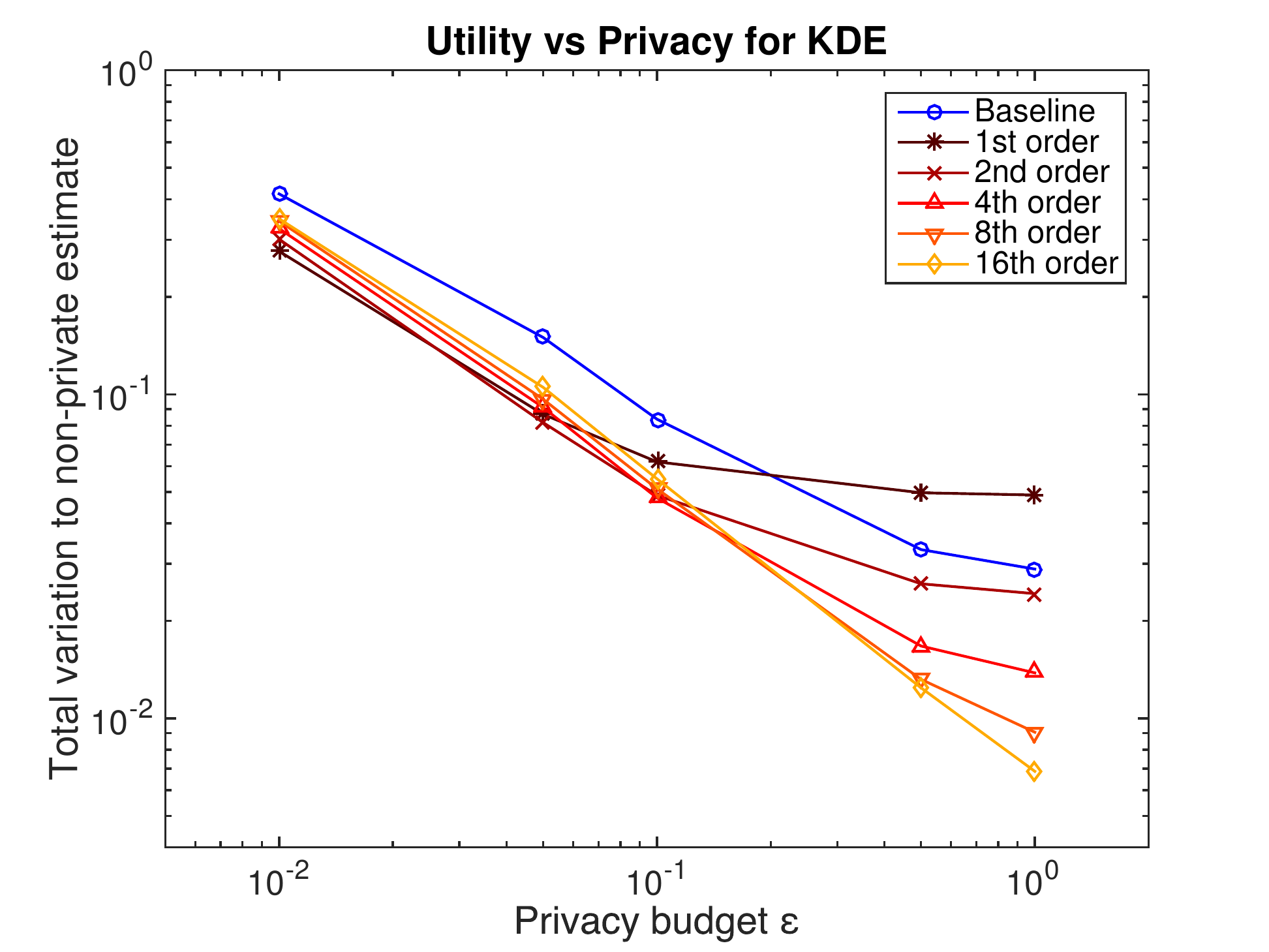}
    \caption{Private KDE with Gaussian kernel}
    \label{fig:kde}
\end{figure}

\paragraph*{Kernel Density Estimation.}

Let $\cX=\cY=[0,1]^{\l}$ and $\cD=(\vect{d}_1,\vect{d}_2,\ldots,\vect{d}_n)\in\cX^{n}$. For a given kernel $K_{H}$, with bandwidth $H$ (a symmetric and positive definite $\l\times\l$ matrix), the kernel density estimator $F_{H}\colon\cX^n\times\cY\ra\reals$ is defined as
$F_{H}(\cD,\vect{y}) = \frac{1}{n}\sum_{i=1}^{n}K_{H}(\vect{y}-\vect{d}_i)$.
It is easy to see that $S(F_{H}) \leq \sup_{\vect{y}\in[-1,1]^{\l}}K_{H}(\vect{y})/n$. For instance, if $K_{H}$ is the Gaussian kernel with covariance matrix $H$, then $S(F_{H}) \leq 1/(n\sqrt{(2\pi)^{\l}\det(H)})$. Moreover, observe that $F_{H}(\cD,\cdot)$ is an $(h,T)$-smooth function for any positive integer $h$. Hence the error introduced by the mechanism is
\[O\left(\frac{1}{n\eps\sqrt{\det(H)}}\log(1/\beta)\right)^{\frac{h}{\l + h}},\]
with probability at least $1-\beta$. In Figure~\ref{fig:kde} we display the utility (averaged over $1000$ repeats) of the Bernstein mechanism ($k=20$) on $5000$ points drawn from a mixture of two normal distributions $N(0.5,0.02)$ and $N(0.75,0.005)$ with weights $0.4,0.6$, respectively. We first observe that for every privacy budget $\eps$ there is a suitable choice of $h$ such that our mechanism always achieves better utility compared to the baseline (\cf Section~\ref{sec:discuss}). Moreover, accuracy improves for increasing $h$, except for sufficiently large perturbations (small $\eps$) which more significantly affect higher-order basis functions (larger $h$).
Private cross validation~\cite{chaudhuri2011differentially,chaudhuri2013stability} can be used to tune $h$. We conclude noting that the same error bounds can be provided by the mechanism of \citenoun{WFZW2013}, since the function $F_{H}(\cD,\cdot)$ is \emph{separable} in the training set $\cD$, \ie $F_{H}(\cD,\cdot) = \sum_{\vect{d}\in\cD}f_{H}(\vect{d},\cdot)$. However, this assumption is overly restrictive for many applications. In the following, we discuss how the Bernstein mechanism can be successfully applied to several such cases.

\paragraph*{Priestley-Chao Kernel Regression.}

For ease of exposition, consider $\l=1$. For constant $B>0$, let $\cX=[0,1]\times[-B,B]$ and $\cY=[0,1]$. Without loss of generality, consider datasets $\cD=((d_1,l_1),(d_2,l_2),\ldots,(d_n,l_n))\in\cX^{n}$, where $d_1\leq d_2\leq\ldots\leq d_n$, and for every $i\in\{1,\ldots,n\}$ there exists $j\neq i$ such that $|d_i-d_j|\leq c/n$, for a given (and publicly known) $0<c=o(n)$. Small values of $c$ restrict the data space under consideration, whereas $c=n$ would correspond to the general case $\cD\in\cX^{n}$. For kernel $K$ and bandwidth $b>0$, the Priestley-Chao kernel estimator~\cite{PC1972,B1977} is defined as
$F_{b}(\cD,y) = \frac{1}{b}\sum_{i=2}^{n}(d_i-d_{i-1})K\left((y-d_i)/b\right)l_i$.
This function is not separable in $\cD$ and
\[S(F_{b}) = \sup_{y\in \cY}S(F_{b}(\cdot,y))\leq \frac{4Bc}{nb} \sup_{y\in[-1,1]} K\left(\frac{y}{b}\right).\]
If $K$ is the Gaussian kernel, then with probability at least $1-\beta$ the error introduced by the mechanism can be bounded by
\[O\left(\frac{c}{n\eps b}\log(1/\beta)\right)^{\frac{h}{1 + h}}.\]

\begin{figure}[t!]
    \centering
    \includegraphics[width=0.9\columnwidth]{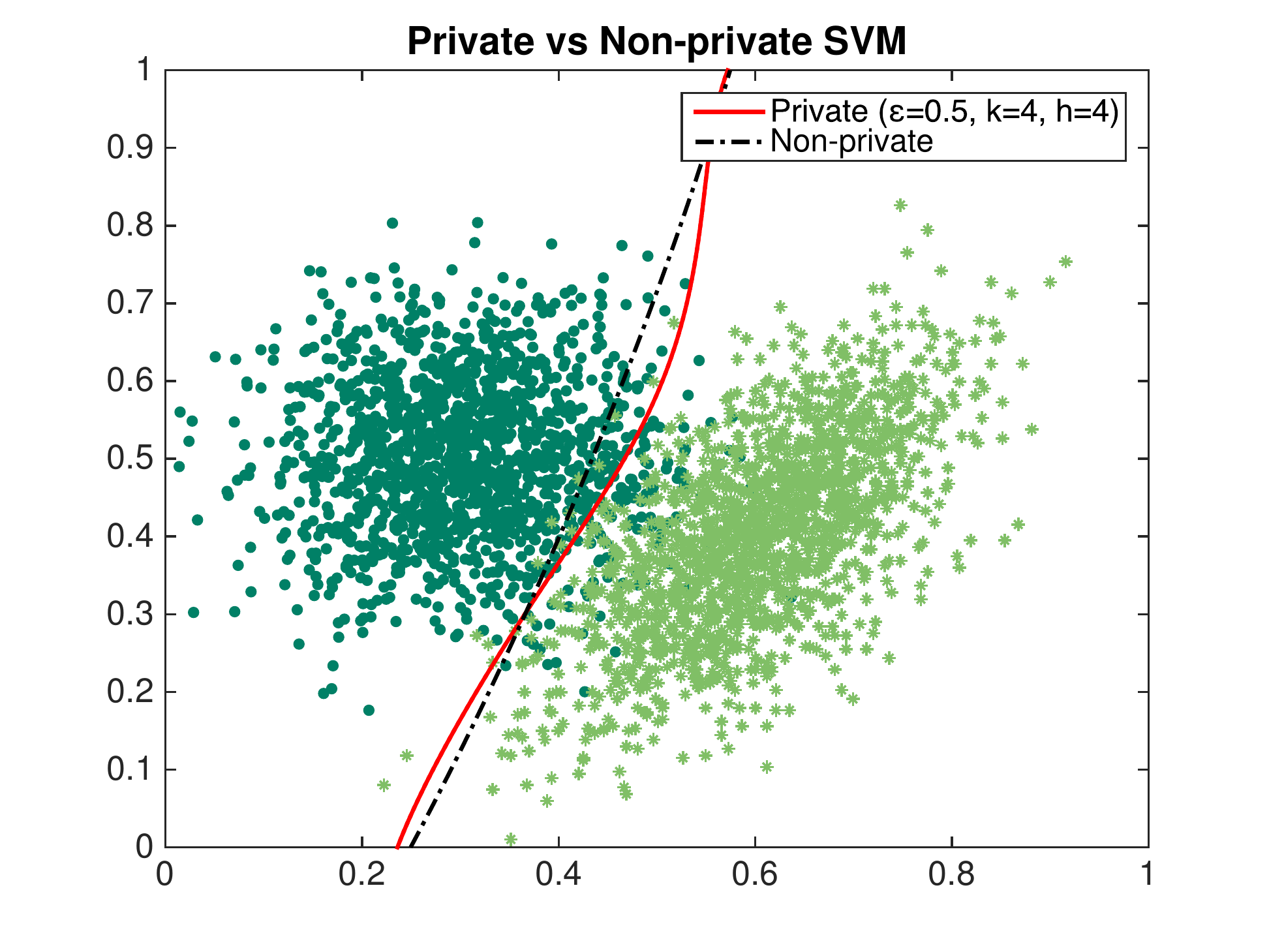}
    \caption{Private SVM with Gaussian kernel}
    \label{fig:svm}
\end{figure}

\paragraph*{Naive Bayes Classification.}

In this example we apply the Bernstein mechanism to a probabilistic learner. Without loss of generality, assume $X=[0,1]^{\l}$, $\cX=X\times\{l^{+},l^{-}\}$, $\cY=X$ and $\cD=((\vect{d}_1,l_1),(\vect{d}_2,l_2),\ldots,(\vect{d}_n,l_n))\in\cX^{n}$. A naive Bayes classifier can be interpreted as $F\colon \cX^{n}\times\cY\ra\reals$ such that $F_{\cD}(\vect{y})\propto\Pr(\vect{y}|l^{+},\cD)\Pr(l^{+}|\cD) - \Pr(\vect{y}|l^{-},\cD)\Pr(l^{-}|\cD)$. Predictions can then be made by assigning the instance $\vect{y}$ to the class $l^{+}$ (resp. $l^{-}$) if $F_{\cD}(\vect{y})\geq 0$ (resp. $F_{\cD}(\vect{y})< 0$). Since, for a class $l$, $\Pr(\vect{y}|l,\cD)\propto\prod_{i=1}^{\l}\Pr(y_i|l,\cD)$, it is easy to show that $F_{\cD}(\cdot)$ is an $(h,T)$-smooth function whenever each likelihood is estimated using a Gaussian distribution or KDE~\cite{JL1995} (with a sufficiently smooth kernel). In the latter case, using a Gaussian kernel, the sensitivity of $F$ can be bounded by $S(F)\leq 2(1/n + (2^{\l}-1)/(n\sqrt{2\pi}b))$, where $b$ is the chosen bandwidth\iffullelse{. The detailed computation is provided in Appendix~\ref{sec:naivebayes}.}{ (\cf the full report~\citenaked{AR2016}).} The error introduced by the Bernstein mechanism is thus bounded by
\[O\left(\frac{1}{n\eps b}\log(1/\beta)\right)^{\frac{h}{\l+h}}\ ,\]
with probability at least $1-\beta$.

\paragraph*{Regularized Empirical Risk Minimization.}

In the next examples, the functions we aim to release are implicitly defined by an algorithm.
Let $X=[0,1]^{\l}$, $\cX=X\times[0,1]$ and $\cY=X$.
Let $L$ be a convex and locally $M$-Lipschitz (in the first argument) loss function. For $\cD=((\vect{d}_1,l_1),(\vect{d}_2,l_2),\ldots,(\vect{d}_n,l_n))\in\cX^{n}$, a regularized empirical risk minimization program with loss function $L$ is defined as
\begin{equation}\label{eq:rERM}
    \vect{w}^{\str}\in\argmin_{\vect{w}\in\reals^{r}} \frac{C}{n}\sum_{i=1}^{n}L(l_i,f_{\vect{w}}(\vect{d}_i)) + \frac{1}{2}\|\vect{w}\|_2^2\ ,
\end{equation}
where $f_{\vect{w}}(\vect{x}) = \langle \phi(\vect{x}),\vect{w}\rangle$ for a chosen feature mapping $\phi\colon X\ra\reals^{r}$ taking points from $X$ to some (possibly infinite) $r$-dimensional feature space and a hyperplane normal $\vect{w}\in\reals^{r}$. Let $K(\vect{x},\vect{y}) = \langle\phi(\vect{x}), \phi(\vect{y})\rangle$ be the kernel function induced by the feature mapping $\phi$. The Representer Theorem~\cite{KW1971} implies that the minimizer $\vect{w}^{\str}$ lies in the span of the functions $K(\cdot,\vect{d}_i)\in\cH$, where $\cH$ is a reproducing kernel Hilbert space (RKHS). Therefore, we consider $F\colon \cX^{n}\times\cY\ra\reals$ such that $F_{\cD}(\vect{y}) = f_{\vect{w}^{\str}}(\vect{y}) = \sum_{i=1}^{n}\alpha_i l_i K(\vect{y},\vect{d}_i)$, for some $\alpha_i\in\reals$. An upper bound on the sensitivity of this function follows from an argument provided by~\citenoun{HRW2013} based on a technique of~\citenoun{BE2002}. In particular, we have
\[S(F) = \sup_{\vect{y}\in \cY,\vect{w}\sim \vect{w}'}\left|f_{\vect{w}}(\vect{y}) - f_{\vect{w}'}(\vect{y})\right| \leq \frac{MC}{n}\sup_{\vect{y}\in\cY}K(\vect{y},\vect{y}).\]
If $K$ is $(2h,T)$-smooth, the error introduced is bounded, with probability at least $1-\beta$, by
\[O\left(\frac{MC\sup_{\vect{y}\in\cY}K(\vect{y},\vect{y})}{n\eps}\log(1/\beta)\right)^{\frac{h}{\l + h}},\]
Note that this result holds with very mild assumptions, namely for any convex and locally $M$-Lipschitz loss function (\eg square-loss, log-loss, hinge-loss) and any bounded kernel $K$. Figure~\ref{fig:svm} depicts SVM learning with RBF kernel ($C=\sigma=1$) on 1500 each of positive (negative) Gaussian data with mean $[0.3,0.5]$ ($[0.6,0.4]$) and covariance $[0.01,0;0,0.01]$ ($0.01*[1,0.8;0.8,1.5]$) and demonstrates the mechanism's uniform approximation of predictions, best seen geometrically with the classifier's decision boundary.

\paragraph*{Logistic Regression.}

Let now $X=\{\vect{x}\in[0,1]^{\l}\colon \|\vect{x}\|_2\leq 1\}$. Let furthermore $\cX=X\times[0,1]$ and $\cY=[0,1]^{\l}$. The logistic regressor can be seen as a function $F\colon \cX^{n}\times\cY\ra\reals$ such that $F_{\cD}(\vect{y}) = \langle \vect{w}^{\str},\vect{y}\rangle$, where $\vect{w}^{\str}$ is the minimizer of~\eqref{eq:rERM} when $\phi$ is the identity mapping and the loss function is $L(l,\langle\vect{w},\vect{d}\rangle)=\log\left(1+e^{-l\langle \vect{w},\vect{d}\rangle}\right)$. It is then possible to show that the error introduced by the Bernstein mechanism is bounded, with probability at least $1-\beta$, by
\[O\left(\frac{C}{n\eps}\log(1/\beta)\right),\]
since $F_{\cD}(\vect{y})$ is a linear function. The prediction with the sigmoid function achieves the same error bound, since it is $1/4$-Lipschitz\iffullelse{. A more detailed analysis is provided in Appendix~\ref{sec:logisticregression}.}{ (\cf the full report~\citenaked{AR2016}).}

\section{Conclusions}\label{sec:conclusions}

In this paper we have considered the release of functions of test data and privacy-sensitive training data. We have presented a simple yet effective mechanism for this general setting, that makes use of iterated Bernstein polynomials to approximate any regular function with perturbations applied to the resulting coefficients. Both $\eps$-differential privacy and utility rates are proved in general for the mechanism, with corresponding lower bounds provided. A number of example learners are analyzed, demonstrating the Bernstein mechanism's versatility.

\subsubsection*{Acknowledgments.}

This work was partially completed while F.~Ald\`a was visiting the University of Melbourne. Moreover, he acknowledges support of the DFG Research Training Group GRK 1817/1. The work of B.~Rubinstein was supported by the Australian Research Council (DE160100584).

\bibliographystyle{aaai}
\bibliography{full_aaai.bbl}

\appendix

\section{Fixed Points of the Bernstein Operator}\label{sec:fixedpoints}

Although this is a classical result, we show for completeness that linear functions are fixed points of the Bernstein operator $B_{k}=B_{k}^{(1)}$, for $k\geq 1$. Let $f(y) = my + q$, for $m,q\in\reals$ and $y\in[0,1]$. We have
\begin{align*}
    B_{k}(f;y) &= \sum_{\nu=0}^{k}f\left(\frac{\nu}{k}\right)b_{\nu,k}(y)\\
    &= \frac{m}{k}\sum_{\nu=0}^{k}\nu b_{\nu,k}(y) + q\sum_{\nu=0}^{k}b_{\nu,k}(y)\\
    &= my + q\ ,
\end{align*}
since $\sum_{\nu=0}^{k}b_{\nu,k}(y)=1$ and $\sum_{\nu=0}^{k}\nu b_{\nu,k}(y)=ky$.

\section{Proof of Theorem~\ref{thm:maintheorem} for $\boldsymbol{\l = 1}$}\label{sec:onedimproof}

Let us fix $k$, a positive integer. As described in Algorithm~\ref{alg:bernstein_mechanism}, the Bernstein mechanism perturbs the evaluation of the function $F_{\cD}$ on a cover of the interval $[0,1]$.

\begin{lemma}\label{lem:dp}
    Let $\eps>0$. Then the Bernstein mechanism $\cM$ provides $\eps$-differential privacy.
\end{lemma}

\begin{proof}
Let $\cD'\in\cX^{n}$ be a second database differing from $\cD$ in one entry only. Let furthermore $\psi\colon \cX^{n}\ra \reals^{k+1}$ be the map defined by
\[\psi(\cD) = \left(F_{\cD}\left(\frac{0}{k}\right),F_{\cD}\left(\frac{1}{k}\right),\ldots,F_{\cD}\left(\frac{k}{k}\right)\right).\]
Then
\begin{align*}
    S(\psi)&=\sup_{\cD\sim\cD'}\|\psi(\cD)-\psi(\cD')\|_1\\
    &\leq \sum_{\nu=0}^{k}\sup_{\cD\sim\cD'}\left|F_{\cD}\left(\frac{\nu}{k}\right)-F_{\cD'}\left(\frac{\nu}{k}\right)\right|\\
    &\leq S(F)(k+1)\ .
\end{align*}
According to Lemma~\ref{lem:sensitivity} (applied with $k+1$ in place of $d$), the mechanism $\cM$ provides $\eps$-differential privacy.
\end{proof}

In order to analyze the accuracy of our mechanism, we denote by
$\widetilde{B^{(h)}_{k}}(F_{\cD};y) = \sum_{\nu=0}^{k}\left[F_{\cD}\left(\nu/k\right)+Z_{\nu}\right]b^{(h)}_{\nu,k}(y)$
the iterated Bernstein polynomial of order $h$ constructed using the coefficients output by the mechanism $\cM$. The error $\alpha$ introduced by the mechanism can be expressed as follows:
\begin{align}
    \alpha &= \max_{y\in[0,1]}\left|F_{\cD}(y) - \widetilde{B^{(h)}_{k}}(F_{\cD};y)\right| \label{eq:error1} \\
    &\leq \max_{y\in[0,1]}\left|\widetilde{B^{(h)}_{k}}(F_{\cD};y) - B^{(h)}_{k}(F_{\cD};y)\right|\label{eq:error2}\\
    \ &\mathrel{\phantom{\leq}}   
    \mathrel + \max_{y\in[0,1]}\left|F_{\cD}(y) - B^{(h)}_{k}(F_{\cD};y)\right|\nonumber.
\end{align}

For every $y\in [0,1]$, the first summand in Equation~\eqref{eq:error2} consists of the absolute value of an affine combination of independent Laplace-distributed random variables. 
\begin{proposition}\label{prop:affinelaplace}
    Let $Z_{0},\ldots,Z_{k}\stackrel{i.i.d.}{\sim} \Lap(\lambda)$, $\tau\geq 0$, and $C_h$ be a constant depending on $h$ only. Then:
    \[
    \Pr\left[\max_{y\in[0,1]}\left|\sum_{\nu=0}^{k} Z_{\nu}b^{(h)}_{\nu,k}(y)\right| \geq \tau\right] \leq
    e^{-\tau/(C_{h}\lambda)}\ .
    \]
\end{proposition}

We provide the proof of Proposition~\ref{prop:affinelaplace} in Appendix~\ref{sec:affinelaplace}.
Proposition~\ref{prop:affinelaplace} implies that with probability at least $1-\beta$ the first summand in Equation~\eqref{eq:error2} is bounded by $O\left(S(F)k\log(1/\beta)/\eps\right)$.
According to the regularity of $F_{\cD}$, the second summand in Equation~\eqref{eq:error2} can then be bounded by a decreasing function $g(k)$. All in all, the error in Equation~\eqref{eq:error1} can be bounded as follows:
\begin{equation}\label{eq:error3}
    \alpha= O\left(g(k) + \frac{S(F)k}{\eps}\log(1/\beta)\right).
\end{equation}
Since the second summand in Equation~\eqref{eq:error3} is an increasing function in $k$, the optimal value for $k$ (up to a constant factor) is achieved when $k$ satisfies
\begin{equation}\label{eq:valuek}
    g(k) = \frac{S(F)k}{\eps}\log(1/\beta)\ .
\end{equation}
Solving Equation~\eqref{eq:valuek} with the bounds for $g(k)$ provided in Theorems~\ref{thm:bernstainapproxsmooth} and~\ref{thm:bernstainapproxhoelder} and substituting the thus obtained value of $k$ into~(\ref{eq:error3}) prove the first two statements. The bound when $F_{\cD}$ is linear follows from the fact that for $h=1$ and $k=1$ the approximation error in Equation~\eqref{eq:error2} is zero. The error is thus bounded by $O\left(S(F)\log(1/\beta)/\eps\right)$. The running time of the mechanism and the running time for answering a query is linear in $k$ and hence upper bounded by a polynomial in $n$ and $1/\eps$, if $1/S(F)\leq\mathrm{poly}(n)$.

\section{Proof of Proposition~\ref{prop:affinelaplace}}\label{sec:affinelaplace}

In order to prove the proposition, we make use of the following result.
\begin{theorem}[\citenaked{P1965}]\label{thm:logconcave}
    Suppose that $f\colon\reals\ra[0,1]$ is a log-concave density function such that $f(y)=f(-y)$ for every $y\in\reals$. Let $Z_1,\ldots,Z_m$ be i.i.d. random variables with density $f$, and suppose that $(a_1,\ldots,a_m),(b_1,\ldots,b_m)\in[0,1]^{m}$ satisfy
    \begin{enumerate}[(i)]
        \item $a_1\geq a_2\geq\ldots\geq a_m$, $b_1\geq b_2\geq\ldots\geq b_m$;
        \item $\sum_{i=1}^{k}b_i\leq\sum_{i=1}^{k}a_i$ for $k=1,\ldots,m-1$;
        \item $\sum_{i=1}^{m}a_i=\sum_{i=1}^{m}b_i=1$.
    \end{enumerate}
    Then, for all $\tau\geq0$
    \[\Pr\left[\left|\sum_{i=1}^{m}b_iZ_i\right|\geq \tau\right] < \Pr\left[\left|\sum_{i=1}^{m}a_iZ_i\right|\geq \tau\right].\]
\end{theorem}

\noindent
Choosing $a_1=1$ and $a_j=0$ for $j=2,\ldots,m$, Theorem~\ref{thm:logconcave} implies
\begin{equation}\label{eq:last}
    \Pr\left[\left|\sum_{i=1}^{m}b_iZ_i\right|\geq \tau\right] < \Pr\left[\left|Z_1\right|\geq \tau\right]
\end{equation}
for every $(b_1,\ldots,b_m)\in[0,1]^{m}$ which satisfies $\sum_{i=1}^{m}b_i=1$. We then observe that the density function $h(y) = \exp(-|y|/\lambda)/(2\lambda)$ of the Laplace distribution is symmetric and log-concave. If $Z_i\sim\Lap(\lambda)$ are i.i.d. random variables for $i=1,\ldots,m$, the right-hand side of Equation~\eqref{eq:last} satisfies
\begin{equation}\label{eq:up_bound_laplace}
    \Pr\left[\left|Z_1\right|\geq \tau\right]=\exp\left(-\frac{\tau}{\lambda}\right).
\end{equation}

Although the bases $b^{(h)}_{\nu,k}$ are not always positive for $h\geq 2$, we observe that, for $y\in[0,1]$, $V(y) = \sum_{\nu=0}^{k} Z_{\nu}b^{(h)}_{\nu,k}(y)$ and $V'(y) = \sum_{\nu=0}^{k} Z_{\nu}|b^{(h)}_{\nu,k}(y)|$ have the same distribution, since the random variables $Z_{\nu}$ are i.i.d. and symmetric around zero. We can thus restrict our analysis to $V'(y)$. For $y\in[0,1]$, let $U(y)=\sum_{\nu=0}^{k}|b^{(h)}_{\nu,k}(y)|$. We first note that
\begin{align}
    U(y) &= \sum_{\nu=0}^{k}\left|\sum_{i=1}^{h}{h \choose i}(-1)^{i-1}B^{i-1}_{k}(b_{\nu,k};y)\right|\nonumber\\
    &\leq \sum_{\nu=0}^{k}\sum_{i=1}^{h}{h \choose i}\left| B^{i-1}_{k}(b_{\nu,k};y)\right|\nonumber\\
    &= \sum_{\nu=0}^{k}\sum_{i=1}^{h}{h \choose i} B^{i-1}_{k}(b_{\nu,k};y)\nonumber\\
    &= \sum_{i=1}^{h}{h \choose i} \sum_{\nu=0}^{k} B^{i-1}_{k}(b_{\nu,k};y)\nonumber\\
    &= \sum_{i=1}^{h}{h \choose i} B^{i-1}_{k}\left(\sum_{\nu=0}^{k}b_{\nu,k};y\right)\nonumber\\
    &= \sum_{i=1}^{h}{h \choose i}\nonumber\\
    &= 2^{h} - 1\ .\label{eq:boundonthesum}
\end{align}

According to Equations~\eqref{eq:last} and~\eqref{eq:up_bound_laplace}, for every $y\in[0,1]$ and $\tau'\geq 0$ we have
\[
\Pr\left[\left|\frac{1}{U(y)}\sum_{\nu=0}^{k} Z_{\nu}|b^{(h)}_{\nu,k}(y)|\right| \geq \tau'\right] \leq
\exp\left(-\frac{\tau'}{\lambda}\right)\ .
\]

\noindent
Choosing $\tau=U(y)\tau'$, we get
\begin{align*}
\Pr\left[\left|\sum_{\nu=0}^{k} Z_{\nu}|b^{(h)}_{\nu,k}(y)|\right| \geq \tau\right] &\leq
\exp\left(-\frac{\tau}{U(y)\lambda}\right)\\
&\leq \exp\left(-\frac{\tau}{(2^{h}-1)\lambda}\right),
\end{align*}
for every $y\in[0,1]$, concluding the proof.

\section{Proof of Proposition~\ref{prop:affinelaplacemultivariate}}\label{sec:affinelaplacemultivariate}

The proof of the proposition follows from the same argument provided in Appendix~\ref{sec:affinelaplace}, with some minor changes. In particular, it suffices to provide a tail bound for
\[\max_{\vect{y}\in[0,1]^{\l}}\left|\sum_{j=1}^{\l}\sum_{\nu_j=0}^{k}Z_{\vect{\nu}}\left|\prod_{i=1}^{\l}{b^{(h)}_{\nu_i,k}(y_i)}\right|\right|,\]
since, as observed in Appendix~\ref{sec:affinelaplace}, the random variables $Z_{\vect{\nu}}$ are i.i.d. and symmetric around zero. In order to apply Theorem~\ref{thm:logconcave} and conclude the proof, we need to upper bound
\[U(\vect{y})=\sum_{j=1}^{\l}\sum_{\nu_j=0}^{k}\left|\prod_{i=1}^{\l}{b^{(h)}_{\nu_i,k}(y_i)}\right|,\]
for every $\vect{y}\in[0,1]^{\l}$. We have
\begin{align*}
    U(\vect{y}) &= \sum_{j=1}^{\l}\sum_{\nu_j=0}^{k}\prod_{i=1}^{\l}\left|{b^{(h)}_{\nu_i,k}(y_i)}\right|\\
    &=\left(\sum_{j=2}^{\l}\sum_{\nu_j=0}^{k}\prod_{i=2}^{\l}\left|{b^{(h)}_{\nu_i,k}(y_i)}\right|\right)\sum_{\nu_1=0}^{k}\left|{b^{(h)}_{\nu_1,k}(y_1)}\right|\\
    &\leq (2^{h} - 1)^{\l}\ ,
\end{align*}
since, according to Equation~\eqref{eq:boundonthesum}, \[\sum_{\nu_{j}=0}^{k}\left|{b^{(h)}_{\nu_{j},k}(y_{j})}\right|\leq (2^{h}-1)\]
for every $j\in\{1,\ldots,\l\}$. The rest of the proof follows from the same computations done at the end of Appendix~\ref{sec:affinelaplace}.

\section{Approximation Error of Multivariate Bernstein Polynomials}\label{sec:induction}

In what follows, we assume that $f\colon[0,1]^{\l}\ra\reals$ is a $(\gamma,L)$-H\"older continuous function. The proof for $(h,T)$-smooth functions follows the same argument, with minor changes. The argument we present here is by induction on $\l$. The base case ($\l=1$) follows from the fact that the Bernstein polynomial $B_{k}(f;y_1)$ converge uniformly to $f$ in the interval $[0,1]$, as shown in Theorem~\ref{thm:bernstainapproxhoelder}. Assume now
\[\left|B_{k}(f;\vect{y})-f(\vect{y})\right|\leq \l L\left(\frac{1}{4k}\right)^{\gamma/2},\]
for every $\vect{y}\in[0,1]^{\l}$. Let $f\colon[0,1]^{\l+1}\ra\reals$ be a $(\gamma,L)$-H\"older continuous function and let $B_{k}(f;\vect{y})$ be the corresponding Bernstein polynomial.
For every $\vect{y}=(y_{1},\ldots,y_{\l+1})\in[0,1]^{\l+1}$, let
\[G(f;\vect{y})=\sum_{j=1}^{\l}\sum_{\nu_j=0}^{k}f\left(\frac{\nu_1}{k},\cdots,\frac{\nu_{\l}}{k},y_{\l+1}\right)\prod_{i=1}^{\l}{b^{(h)}_{\nu_i,k}(y_i)}.\]
The error $\left|B_{k}(f;\vect{y})-f(\vect{y})\right|$ can then be bounded by
\begin{align}
    &\leq\left|B_{k}(f;\vect{y})-G(f;\vect{y})\right|+ \left|G(f;\vect{y})-f(\vect{y})\right|\label{eq:inductbernstein}\\
    &\leq L\left(\frac{1}{4k}\right)^{\gamma/2} + \l L\left(\frac{1}{4k}\right)^{\gamma/2}\nonumber\\
    &= (\l+1)L\left(\frac{1}{4k}\right)^{\gamma/2}.\nonumber
\end{align}

In fact, the second term of Equation~\eqref{eq:inductbernstein} is the error of the Bernstein polynomial of $f$ seen as a function of $y_1,\ldots,y_{\l}$ only. The corresponding bound then follows from the inductive step.
On the other hand, the first summand corresponds to the approximation error of the (univariate) Bernstein polynomial of $G(f,\vect{y})$ as a function of the remaining variable $y_{\l+1}$.
The statement for $(h,T)$-smooth functions is similarly obtained by replacing $B_{k}$ with $B^{(h)}_{k}$ and using the bound of Theorem~\ref{thm:bernstainapproxsmooth}.

\section{Naive Bayes Classification}\label{sec:naivebayes}

In this section, we show how to bound the sensitivity of a naive Bayes classifier $F_{\cD}$, as defined in Section~\ref{sec:examples}.
\begin{align*}
    S(F) &= \sup_{\vect{y}\in[0,1]^{\l}, \cD\sim\cD'}|F(\cD,\vect{y})-F(\cD',\vect{y})|\\
    &\leq 2\sup_{l\in\{l^{+},l^{-}\}, \vect{y}\in[0,1]^{\l}, \cD\sim\cD'}|\Pr(\vect{y}|l,\cD)\Pr(l|\cD)\\
    \ &\mathrel{\phantom{\leq\ }}   
    \mathrel -\Pr(\vect{y}|l,\cD')\Pr(l|\cD')|\ .
\end{align*}

We assume that a class probability $\Pr(l|\cD)$ is estimated using the corresponding relative frequency in the training set~$\cD$. Therefore, for $\cD\sim\cD'$, $\Pr(l|\cD)\leq\Pr(l|\cD')+1/n$. Assume now that for every $\vect{y}\in[0,1]^{\l}$, $\cD\sim\cD'\in\cX^{n}$, $l\in\{l^{+},l^{-}\}$ and $i\in\{1,\ldots,\l\}$ there exists $0\leq\xi<1$ such that
\[\left|\Pr(y_i|l,\cD)-\Pr(y_i|l,\cD')\right|\leq\xi\]
holds. We then have
\begin{align*}
    \Pr(\vect{y}|l,\cD) &\propto \prod_{i=1}^{\l}\Pr(y_i|l,\cD)\\
    &\leq \prod_{i=1}^{\l}\left(\Pr(y_i|l,\cD')+\xi\right)\\
    &\leq \prod_{i=1}^{\l}\Pr(y_i|l,\cD') + (2^{\l}-1)\xi\ ,
\end{align*}
where the last inequality follows from the fact that there are $2^{\l}-1$ cross products and each one of them has at least a $\xi$ factor. If each (unidimensional) likelihood is estimated using KDE~\cite{JL1995} with a Gaussian kernel of bandwidth $b$, $\xi$ corresponds to the upper bound on the sensitivity of KDE shown in Section~\ref{sec:examples}. Putting all the pieces together, we obtain
\begin{align*}
    S(F) &\leq 2\sup_{l\in\{l^{+},l^{-}\},\vect{y}\in[0,1]^{\l},\cD\sim\cD'} \Pr(l|\cD')\frac{2^{\l}-1}{n\sqrt{2\pi}b}\\
    \ &\mathrel{\phantom{\leq\ }}   
    \mathrel + \frac{1}{n}\Pr(\vect{y}|l,\cD)\\
    &\leq 2\left(\frac{2^{\l}-1}{n\sqrt{2\pi}b} + \frac{1}{n}\right).
\end{align*}

\section{Logistic Regression}\label{sec:logisticregression}

Let $X=\{\vect{x}\in[0,1]^{\l}\colon \|\vect{x}\|_2\leq 1\}$. Let furthermore $\cX=X\times[0,1]$ and $\cY=[0,1]^{\l}$. The logistic regressor can be seen as a function $F\colon \cX^{n}\times\cY\ra[0,1]$ such that, for $\cD=((\vect{d}_1,l_1),(\vect{d}_2,l_2),\ldots,(\vect{d}_n,l_n))\in\cX^{n}$, $F_{\cD}(\vect{y}) = 1/(1+\exp(-\langle \vect{w}^{\str},\vect{y}\rangle))$, where $\vect{w}^{\str}$ is such that
\[\vect{w}^{\str}\in\argmin_{\vect{w}\in\reals^{\l}} \frac{C}{n}\sum_{i=1}^{n}\log\left(1+e^{-l_i\langle \vect{w},\vect{d}_i\rangle}\right) + \frac{1}{2}\|\vect{w}\|_2^2\ .\]

In order to compute $S(F)$, we first observe that the sigmoid function is $1/4$-Lipschitz. Denoting by $\vect{w}\sim \vect{w}'$ the minimizers obtained from input databases $\cD\sim\cD'$, we have
\begin{align*}
    S(F)&\leq \sup_{\vect{y}\in \cY,\vect{w}\sim \vect{w}'} \frac{1}{4}|\langle \vect{w}-\vect{w}',\vect{y}\rangle|\\
    &\leq \sup_{\vect{y}\in \cY,\vect{w}\sim \vect{w}'} \frac{1}{4}\|\vect{w}-\vect{w}'\|_2 \|\vect{y}\|_2\ ,
\end{align*}
where the last inequality follows from an application of the Cauchy-Schwarz inequality.~\citenoun{CM2008} showed $\sup_{\vect{w}\sim \vect{w}'}\|\vect{w}-\vect{w}'\|_2\leq 2C/n$. Since $\|\vect{y}\|_2\leq \sqrt{\l}$ for every $\vect{y}\in\cY$, we have $S(F)\leq C\sqrt{\l}/(2n)$. Since $F_{\cD}$ is an $(h,T)$-smooth function for any positive integer $h$, with probability at least $1-\beta$ the error introduced by the mechanism is bounded by
\[O\left(\frac{C}{n\eps}\log(1/\beta)\right)^{\frac{h}{\l + h}}.\]
We note that defining $F_{\cD}(\vect{y})=\langle \vect{w}^{\str},\vect{y}\rangle$ the previous bound can be improved to
\[O\left(\frac{C}{n\eps}\log(1/\beta)\right),\]
since $S(F)\leq 2C\sqrt{\l}/n$ and $F_{\cD}(\vect{y})$ is a linear function. The prediction with the sigmoid function achieves the same error bound, being $1/4$-Lipschitz.

\end{document}